\documentclass[]{llncs}
\usepackage{graphicx}
\usepackage{floatflt, wrapfig}
\usepackage{xcolor}
\usepackage{ amsmath, mathtools, mathptmx}
\usepackage{amssymb, amsbsy, bm, fixmath}
\usepackage{graphicx}
\graphicspath{ {images/} }
\usepackage{epstopdf}
\usepackage{color, import}
\usepackage{sectsty}
\usepackage{enumitem}
\usepackage{lineno}
\usepackage[hidelinks]{hyperref}
\usepackage[edges]{forest}
\usepackage[linesnumbered,vlined,ruled]{algorithm2e}
\usepackage{setspace}
\usepackage[list=true]{subcaption}
\usepackage[title]{appendix}
\usepackage{multirow}
\usepackage{array}
\usepackage{tikz}
\usetikzlibrary{positioning,arrows,automata}

\begin{document}
\title{Arbitrary Pattern Formation on Infinite Grid by Asynchronous Oblivious Robots \thanks{This is the full version of the paper, with the same title and authors, that was accepted in the 13th International Conference and Workshops on Algorithms and Computation (WALCOM 2019), February 27 - March 02, 2019, Guwahati, India}}
%
%
\author{Kaustav Bose\orcidID{0000-0003-3579-1941} \and
Ranendu	Adhikary\orcidID{0000-0002-9473-2645} \and
Manash Kumar Kundu\orcidID{0000-0003-4179-8293} \and
Buddhadeb Sau}
%
%
\institute{Department of Mathematics, Jadavpur University, Kolkata, India\\
\email{\{kaustavbose.rs, ranenduadhikary.rs, manashkrkundu.rs \}@jadavpuruniversity.in, bsau@math.jdvu.ac.in}
}
\maketitle                 

\begin{abstract}

The \textsc{Arbitrary Pattern Formation} problem asks to design a distributed algorithm that allows a set of autonomous mobile robots to form any specific but arbitrary geometric pattern given as input. The problem has been extensively studied in literature in continuous domains. This paper investigates a discrete version of the problem where the robots are operating on a two dimensional infinite grid. The robots are assumed to be autonomous, identical, anonymous and oblivious. They operate in Look-Compute-Move cycles under a fully asynchronous scheduler. The robots do not agree on any common global coordinate system or chirality. We have shown that a set of robots can form any arbitrary pattern, if their starting configuration is asymmetric.


\end{abstract}
\section{Introduction}

\subsection{Motivation}

Distributed coordination of autonomous mobile robot systems has attracted considerable research interest in recent years owing to its potential applications in a wide range of real-world problems. The problem of forming an arbitrary geometric pattern is a fundamental coordination task for autonomous robot swarms. The pattern formation problem has been extensively investigated in continuous domains under different assumptions. In the continuous setting, the robots are assumed to be able to execute accurate movements in arbitrary directions and by arbitrarily small amounts. Hence, even in densely crowded situations, the robots can maneuver avoiding collisions. Certain models also permit the robots to move along curved trajectories, in particular, circumference of a circle. For robots with weak mechanical capabilities, it may not be possible to execute such intricate movements with precision. This motivates us to consider the problem in a grid based terrain where the movements of the robots are restricted only along grid lines and only by a unit distance in each step. Grid type floor layouts can be easily implemented in real life robot navigation systems using magnets or optical guidances.

\subsection{Earlier Works}

The \textsc{Arbitrary Pattern Formation} problem was first studied by Suzuki and Yamashita \cite{Yamashita96,Yamashita99} on the Euclidean plane. In these papers, a complete characterization of the class of formable patterns has been provided for autonomous and anonymous robots with an unbounded amount of memory. They characterized the class of formable patterns by using the notion of \emph{symmetricity} which is essentially the order of the cyclic group that acts on the initial configuration. In \cite{Flocchini08}, Flocchini et. al. investigated the solvability of the problem for fully asynchronous and oblivious robots. Initially, the robots are in arbitrary positions, with the only requirement that no two robots are in the same position. They showed that if the robots have no common agreement on coordinate system, they cannot form an arbitrary pattern. If the robots have one-axis agreement, then any odd number of robots can form an arbitrary pattern, but an even number of robots cannot, in the worst case. If the robots agree on both $X$ and $Y$ axes, then any set of robots can form any pattern. They also proved that it is possible to elect a leader for $n \geq 3$ robots if it is possible to form any pattern. In \cite{Petit09,Dieudonne10}, the authors studied the relationship between \textsc{Arbitrary Pattern Formation} and \textsc{Leader election} among robots in asynchronous scheduler. They provided algorithms to form an arbitrary pattern starting from any geometric configuration wherein the leader election is possible. More precisely, their solutions work for four or more robots with chirality and for at least five robots without chirality. Combined with the result in \cite{Flocchini08}, they deduced that \textsc{Arbitrary Pattern Formation} and \textsc{Leader election} are equivalent, i.e., it is possible to solve \textsc{Arbitrary Pattern Formation} for $n \geq 4$ with chirality (resp. $n \geq 5$ without chirality) if and only if \textsc{Leader election} is solvable. While all the previous works considered robots with unlimited visibility,  Yamauchi et. al. \cite{Yamauchi13} first studied the problem with limited visibility. Randomized pattern formation algorithms were studied in \cite{Bramas16,Yamauchi14}. In \cite{Das15}, Das et al. investigated the problem of forming a sequence of patterns in a given order. In \cite{Cicerone17,Fujinaga15}, the problem was studied allowing the pattern to have multiplicities. In \cite{cicerone2018embedded,Fujinaga15} the so-called \textsc{Embedded Pattern Formation} problem was studied where the pattern to be formed is provided as a set of visible points in the plane. Recently in \cite{Lukovszki14}, the pattern formation problem was studied on an infinite grid for robots with limited visibility. The problem was studied in synchronous setting for robots with constant size memory, and having a common coordinate system. Furthermore, robots were given a fixed point on the grid so that they can form a connected configuration containing it. Other specific types of formation problems that have been studied in the infinite grid set up, are the \emph{Gathering} problem \cite{Stefano17}, i.e., the point formation problem and the \emph{Mutual Visibility }problem \cite{Adhikary18}, where a set of opaque robots have to form a pattern in which no three robots are collinear.

The paper is organized as follows. In Section \ref{sec2}, some basic definitions and a formal description of the model and the problem are presented. In Section \ref{sec3}, we present and solve a preliminary problem that will be used in the main algorithm, which is described with formal proof in Section \ref{secmain}.

\section{Model and Definition}\label{sec2}

\subsection{The Model}\label{assum}

\quad \enspace \textbf{Robots:} The robots are \emph{autonomous} (there is no central control), \emph{homogeneous} (they execute the same distributed algorithm), \emph{anonymous} (they have no unique identifiers), \emph{identical} (they are indistinguishable by their appearance) and \emph{oblivious} (they have no memory of past configurations and previous actions). The robots cannot explicitly communicate with each other. The robots have \emph{global visibility} which means that they can observe the entire grid and the positions of all the robots. The robots do not have access to any common global coordinate system. In particular, they do not have a common notion of direction or chirality. Each robot has its own local view of the world with respect to its local Cartesian coordinate system. All the robots are initially positioned on distinct grid points.

\textbf{Movement:} The movement of the robots are restricted only along grid lines from one grid point to one of its four neighboring grid points. Traditionally in discrete domains, robot movements are assumed to be instantaneous. For simplicity of analysis, we also assume the movements to be instantaneous. This implies that the robots are always seen on grid points, not on edges. However, our strategy will also work without this assumption (by asking the robots to wait i.e, do nothing, if they see a robot on a grid edge).

\textbf{Look-Compute-Move cycles:} The robots, when active, operate according to the so-called \textsc{Look-Compute-Move} cycle. In each cycle, a previously idle or inactive robot wakes up and executes
the following steps. In the \textsc{Look} phase, the robot takes the snapshot of the positions of all the robots, represented in its own local coordinate system. Based on the perceived configuration, the robot performs computations according to a deterministic algorithm to decide whether to stay put or to move to an adjacent grid point. Based on the outcome of the algorithm, the robot either remains stationary or makes an instantaneous move to an adjacent grid point.

\textbf{Scheduler:} We assume that the robots are controlled by a fully asynchronous adversarial scheduler (\emph{ASYNC}). This implies that the amount of time spent in \textsc{Look}, \textsc{Compute}, \textsc{Move} and inactive states by different robots is finite but unbounded and unpredictable. As a result, the robots do not have a common notion of time and the configuration perceived by a robot during the \textsc{Look} phase may significantly change before it actually makes a move.

\subsection{Basic Geometric Definitions}

Consider a team of a finite number of robots placed on the vertices of a simple undirected connected graph $G=(V,E)$. Define a function $f:V \longrightarrow \mathbb{N} \cup \{0\}$, where $f(v)$ is the number of robots on the vertex $v$\footnote{Since we have assumed that the robots are initially positioned on distinct grid points and our algorithm guarantees collisionless movements, $f(v)$ is always either $0$ or $1$.}. The pair $(G,f)$ is called a \emph{configuration of robots on} $G$, or simply a \emph{configuration}. Given a configuration of robots $\mathcal{C}$, let $\mathcal{R}$ denote the smallest grid-aligned rectangle that contains all the robots.

An \emph{automorphism} of a graph $G = (V, E)$ is a bijection $\varphi : V \longrightarrow V$ such that for all $u, v \in V$, $u, v$ are adjacent if and only if $\varphi(u), \varphi(v)$ are adjacent. The set of all automorphisms of $G$ forms a group, called the \emph{automorphism group} of G and is denoted by $Aut(G)$. The definition of automorphism of graphs can be extended to robot configurations on graphs. An \emph{automorphism of a configuration} $(G,{f})$ is an automorphism $\varphi$ of $G$ such that ${f}(v)= {f}(\varphi(v))$ for all $v \in V$. The set of all automorphisms of $(G,{f})$ also forms a group that will be denoted by $Aut(G,{f})$. We shall refer to an automorphism of a configuration as a \emph{symmetry}. We shall call a symmetry \emph{trivial} if $\varphi(v) = v$, for all $v \in V$ with $f(v) \neq 0$. If a configuration admits no non-trivial symmetries, then it is called an \emph{asymmetric configuration}, and otherwise, a \emph{symmetric configuration}.

An infinite path is the graph $P = (\mathbb{Z}, E)$, where $E = \{ (i , i + 1) \mid i \in \mathbb{Z}\}$. An infinite grid can be defined as the Cartesian product $G = P \times P$. Assume that the infinite grid $G$ is embedded in the Cartesian plane $\mathbb{R}^2$. It is not difficult to see that $Aut(G)$ is generated by three types of automorphisms: translations, reflections and rotations. A translation shifts all the vertices of $G$ by the same amount. Since a configuration $(G,{f})$ has only finite number of robots, it is not difficult to see that $Aut(G,{f})$ has no translations. Reflections are defined by an axis of reflection. The axis can be horizontal or vertical or diagonal. The angle of rotation can be of 90 or 180 degrees, and the center of a rotation can be a vertex, or the center of an edge, or the center of the unit square.

The solvability of the arbitrary pattern formation problem depends on the symmetries of the initial configuration of the robots. This paper exclusively considers only asymmetric initial configurations. Some impossibility results regarding symmetric configurations are briefly discussed in Section \ref{conclu}.

\subsection{The Arbitrary Pattern Formation Problem}

A swarm of $k$ robots is arbitrarily deployed on the vertices of the infinite grid. We assume that the initial configuration $\mathcal{C}_{init}$ is asymmetric, and no two robots are in the same position. The goal of the \textsc{Arbitrary Pattern Formation} problem is to design a distributed algorithm that guides the robots to form an arbitrary geometric pattern $\mathcal{C}_{target}$. The pattern $\mathcal{C}_{target}$ is a set of $k$ (distinct) vertices in the grid given in an arbitrary Cartesian coordinate system. The pattern  $\mathcal{C}_{target}$ is given to all robots in the system as input. Due to absence of a common global coordinate system, the robots decide that the pattern is formed
when their present configuration becomes `similar' to $\mathcal{C}_{target}$ with respect to translations, rotations,
reflections. We say that a pattern formation algorithm is \emph{collision-free}, if, at any time $t$, there are no two robots that occupy the same grid point. Avoiding collisions is a necessary requirement of the problem under this model. This is because, if two robots at any point in time, occupy the same grid point, they can not be deterministically separated thereafter, as they both execute the same deterministic algorithm.

\section{Pattern Formation on a Finite Grid}\label{sec3}

In this section, we will discuss a related problem that will be used in the main algorithm. Consider a set of $k$ robots deployed on an $m \times n$ finite grid. Starting from any arbitrary (symmetric or asymmetric) configuration, they are required to form a given arbitrary pattern. Unlike our original problem, we assume that the robots agree on a common global coordinate system. The input $\mathcal{C}_{target}$ is also given in this coordinate system. Hence, the given input corresponds to a fixed set $\mathcal{T}$ of $k$ grid points on the grid and our problem is to place a robot on each of these grid points. All the other assumptions from our original problem, stated in Section \ref{assum}, are retained. 

We shall first consider the case where $m = 1$, i.e., the grid is just a discretized line segment. Since the robots have a common global coordinate system, they agree on left and right. Hence, in the starting configuration, the robots can be labeled as $r_1, \ldots, r_k$ from left to right. If we can devise a swap-free (the act of two adjacent robots exchanging their positions is called a swap) and collision-free movement strategy, then the labels will remain unchanged throughout the algorithm. Note that, in the asynchronous setting, a collision-free algorithm is necessarily swap-free. We can also label the grid points in $\mathcal{T}$ as $t_1, \ldots, t_k$ from left to right. Our strategy is to simply ask each $r_i$ to go to $t_i$. In order to avoid collisions, a robot will move to an adjacent grid point only if it is empty. A pseudocode description of the strategy is given in Algorithm \ref{gather}. 

\begin{figure}[h]
  \centering
  \begin{minipage}{.7\linewidth}
  \begin{algorithm}[H]
    \setstretch{0.2}
    \SetKwInOut{Input}{Input}
    \SetKwInOut{Output}{Output}
    \SetKwProg{Fn}{Function}{}{}
    \SetKwProg{Pr}{Procedure}{}{}

    \Pr{\textsc{PFonPath()}}{
    
    $s \in \{left, right\}$
    
    $r_i \leftarrow$ me
    
    \If{I am not at $t_i$}{
    
	\If{$t_i$ is on my $s$}{
	
	    $u \leftarrow$ the adjacent grid point on my $s$
	    
	    \If{$u$ is empty}{Move to $u$}
	
	    }
      }
    }

\caption{\textbf{Pattern formation on a $1 \times n$ grid}}
    \label{gather} 
\end{algorithm}

\end{minipage}
\end{figure}

\begin{theorem}\label{lineth}
 Algorithm \textsc{PFonPath()} is correct.
\end{theorem}

\begin{proof}
 
 We first show that the algorithm does not lead to a collision. A collision can only occur in a situation shown in Fig. \ref{proof}, where the algorithm asks both $r_{i}$ and $r_{i+1}$ to move to $B$. This implies that $t_{i}$ is somewhere on the right of $A$, while $t_{i+1}$ is on the left of $C$. But this is impossible, since $t_{i}$ and $t_{i+1}$ are two distinct grid points with $t_{i+1}$ on the right of $t_{i}$.
 
 It remains to show that following this strategy, each $r_i$ will be able to reach $t_{i}$. The algorithm can only fail if a deadlock is created. This can only happen if we have two robots, $r_{i}$ and $r_{i+1}$, adjacent to each other, where $r_{i}$ wants to move towards right and $r_{i+1}$ is either at $t_{i+1}$ or wants to move towards left. Again this is impossible by the same arguments as earlier. \qed

\end{proof}

  \begin{figure}[h]
\centering
\subcaptionbox[Short Subcaption]{
      \label{proof}
}
[
    0.49\textwidth 
]
{
    \fontsize{9pt}{9pt}\selectfont
    \def\svgwidth{0.40\textwidth}
    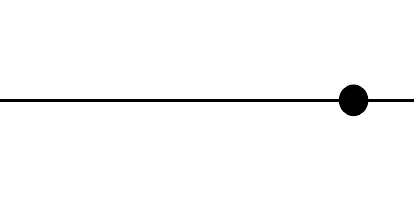
}
\hfill 
\subcaptionbox[Short Subcaption]{
     \label{coil}
}
[
    0.40\textwidth 
]
{
    \fontsize{9pt}{9pt}\selectfont
    \def\svgwidth{0.40\textwidth}
    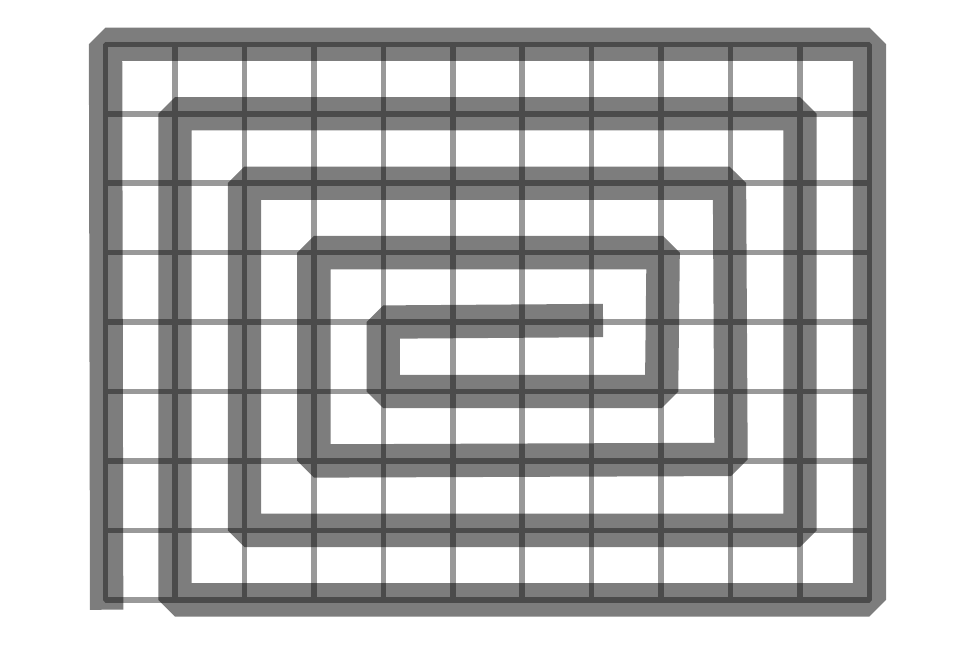
}

\caption[Short Caption]{a) Illustration supporting the proof of Theorem \ref{lineth}. b) A coiled up path in a finite grid.}

\end{figure}

 Now we consider the general $m \times n$ finite grid. An $m \times n$ finite grid can be seen as a coiled up path as shown in Fig. \ref{coil}. To be precise, an $m \times n$ grid has a spanning subgraph isomorphic to the finite path $P_{mn}$. But there are many such spanning subgraphs. The common global coordinate system allows the robots to agree on a particular subgraph as shown in Fig. \ref{coil}. Hence, pattern formation on a finite grid reduces to pattern formation on a path, which can be solved by algorithm \textsc{PFonPath()}.

\section{The Main Algorithm}\label{secmain}

\begin{floatingfigure}[r]{6cm}
   
   \fontsize{9pt}{9pt}\selectfont
   \def\svgwidth{0.4\textwidth}
\begingroup%
  \makeatletter%
  \providecommand\color[2][]{%
    \errmessage{(Inkscape) Color is used for the text in Inkscape, but the package 'color.sty' is not loaded}%
    \renewcommand\color[2][]{}%
  }%
  \providecommand\transparent[1]{%
    \errmessage{(Inkscape) Transparency is used (non-zero) for the text in Inkscape, but the package 'transparent.sty' is not loaded}%
    \renewcommand\transparent[1]{}%
  }%
  \providecommand\rotatebox[2]{#2}%
  \ifx\svgwidth\undefined%
    \setlength{\unitlength}{201.25984252bp}%
    \ifx\svgscale\undefined%
      \relax%
    \else%
      \setlength{\unitlength}{\unitlength * \real{\svgscale}}%
    \fi%
  \else%
    \setlength{\unitlength}{\svgwidth}%
  \fi%
  \global\let\svgwidth\undefined%
  \global\let\svgscale\undefined%
  \makeatother%
  \begin{picture}(1,0.85915493)%
    \put(0,0){\includegraphics[width=\unitlength,page=1]{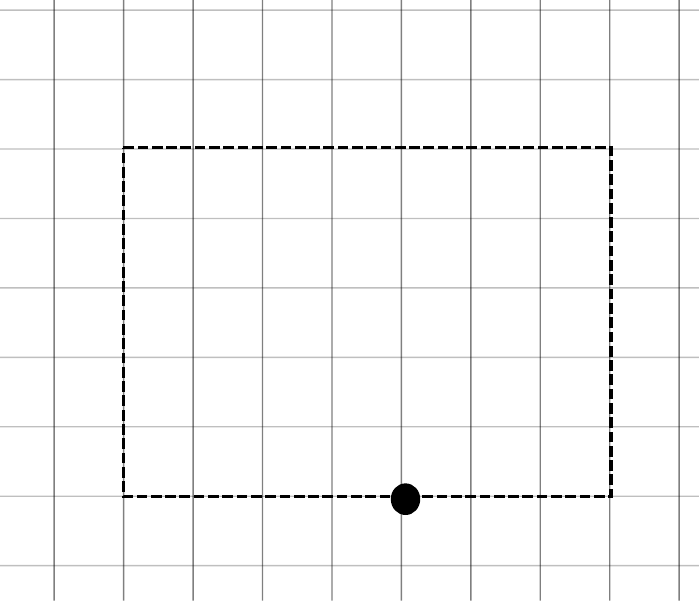}}%
    \put(0.0383405,0.29013368){\color[rgb]{0,0,0}\makebox(0,0)[lb]{\smash{$r_1$}}}%
    \put(0.88256938,0.28872925){\color[rgb]{0,0,0}\makebox(0,0)[lb]{\smash{$r_2$}}}%
    \put(0,0){\includegraphics[width=\unitlength,page=2]{lex_example.pdf}}%
    \put(0.09924162,0.67097858){\color[rgb]{0,0,0}\makebox(0,0)[lb]{\smash{$D$}}}%
    \put(0.87274311,0.68011036){\color[rgb]{0,0,0}\makebox(0,0)[lb]{\smash{$C$}}}%
    \put(0.87941435,0.09526253){\color[rgb]{0,0,0}\makebox(0,0)[lb]{\smash{$B$}}}%
    \put(0.098876,0.10226735){\color[rgb]{0,0,0}\makebox(0,0)[lb]{\smash{$A$}}}%
  \end{picture}%
\endgroup%

   \caption{In this configuration, the lexicographically largest string is $\lambda_{AD} =$ $0101000010001000000001011000010100$ $00000010001000$. The head and the tail are respectively $r_1$ and $r_2$. }
   \label{lex_example}
   \end{floatingfigure}

Consider a configuration $\mathcal{C}$, where $\mathcal{R} = ABCD$ is an $m \times n$ rectangle with $|AB| = n \geq m = |AD|$ (See Fig. \ref{lex_example}). Here, the size of a side is defined as the number of grid points on it. If all robots in $\mathcal{C}$ lie on one grid line, then $\mathcal{R}$ is just a line segment. In this case, when we say $\mathcal{R} = ABCD$, it is to be regarded as a $1 \times n$ `rectangle' with $A=D$, $B=C$ and $|AD| = |BC| = 1$. Let us first assume that $ABCD$ is a non-square rectangle with $|AB| = n > m = |AD| > 1$. We associate a binary string of length $mn$ to each corner of $\mathcal{R}$. The binary string associated to a corner $A$ is defined as follows. Scan the grid from $A$ along the shorter side $AD$ to $D$ and sequentially all grid lines parallel to $AD$ in the same direction. For each grid point, put a $0$ or $1$ according to whether it is empty or occupied. We denote this string by $\lambda_{AD}$. The three other strings $\lambda_{BC}$, $\lambda_{CB}$ and $\lambda_{DA}$ are defined similarly. If $ABCD$ is a square, i.e., $m = n$, then we have to associate two strings to each corner. In that case, the two sequences associated with $A$ will be denoted by $\lambda_{AD}$ and $\lambda_{AB}$. If any two of these strings are equal, then it implies that the configuration has a (reflectional or rotational) symmetry. Hence, if the configuration is asymmetric, then all the strings are distinct and we can find a unique lexicographically largest string. Assume that $\lambda_{AD}$ is the lexicographically largest string. Then $A$ will be called the \emph{leading corner}. Once we have the unique lexicographically largest string $\lambda_{AD}$, the robots can agree on a common coordinate system as follows. The leading corner $A$ is taken as origin and $X$-axis $= \overrightarrow{AB}$, $Y$-axis $= \overrightarrow{AD}$. Unless mentioned otherwise, any asymmetric configuration $\mathcal{C}$ will be expressed in this coordinate system. In the case where $m=1$,  $\lambda_{AD}$ and $\lambda_{DA}$ essentially refers to the same string. Hence, in this case, we have only two binary strings to compare. Again, if they are equal then the configuration is symmetric. Hence, if the configuration is asymmetric \footnote{In the $m=1$ case, we already have a reflectional symmetry with respect to $\overleftrightarrow{AB}$. But this is a trivial symmetry, and is to be ignored by our definition of asymmetric configurations.}, then we shall have a leading corner, say $A = D$. Then $A$ will be taken as origin and $X$-axis $= \overrightarrow{AB}$. But there will be no agreement on the $Y$-axis. However, as all the robots lie on the $X$-axis, the points in $\mathcal{C}$ can still be unambiguously expressed in coordinates.

Therefore, we see that in an asymmetric configuration, all the robots can agree on a common global coordinate system. By `up' (resp. `right') and `down' (resp. `left'), we shall refer to the positive and negative direction of $Y$ (resp. $X$) axis of this coordinate system respectively. Also, given any asymmetric configuration $\mathcal{C}$, the robots corresponding to the first and the last 1 in the lexicographically largest string, will be called the \emph{head} and the \emph{tail} respectively. The remaining robots will be called \emph{interior} robots. $\mathcal{C}'$ and $\mathcal{C}''$ will denote the sets $\mathcal{C} \setminus \{$tail$\}$ and $\mathcal{C} \setminus \{$head, tail$\}$ respectively.

The configuration $\mathcal{C}_{target}$, given to the robots as an input, is expressed in some arbitrary coordinate system. We can take the smallest enclosing rectangle of $\mathcal{C}_{target}$, call it $\mathcal{R}_{target}$. Assume that $\mathcal{R}_{target}$ is an $M \times N$ rectangle, with $N \geq M$. Now associate binary strings to its corners in the same manner as we did for $\mathcal{R}$. We shall assume that $\mathcal{C}_{target}$ is expressed in a coordinate system where the origin is the leading corner and  the positive $Y$ axis is along the side corresponding to the lexicographically largest string. No generality is lost, as the robots can always perform such a coordinate transformation on the input. However, unlike the previous case, we may not have a unique lexicographically largest string. This is because, the configuration $\mathcal{C}_{target}$ can have symmetries. In that case, for the coordinate transformation, we have to choose one among the largest strings to define the coordinate system. Notice that any choice leads to the same set of values. Therefore, in general, we shall assume that the origin of the coordinate system of $\mathcal{C}_{target}$ is \emph{one of the} leading corners, and the positive $Y$ axis is along the side corresponding to \emph{one of the} lexicographically largest strings. We shall call this coordinate system the \emph{canonical coordinate system}. Given $\mathcal{C}_{target}$ in the canonical coordinate system, we define $h_{target}, t_{target} \in \mathcal{C}_{target}$ as the points, corresponding to the first and the last 1 of the binary string that starts from the origin and goes along the $Y$ axis respectively. Also, define $\mathcal{C}_{target}'$ $= \mathcal{C}_{target} \setminus \{t_{target}\}$ and $\mathcal{C}_{target}''$ $= \mathcal{C} \setminus \{h_{target}, t_{target}\}$.

\begin{figure}

\begin{tabular}{| m{1cm} | m{11cm}|} 
\hline
$C_0$ & $\mathcal{C} = \mathcal{C}_{target}$ \\ 
\hline
$C_1$ & $\mathcal{C'} = \mathcal{C'}_{target}$  \\ 
\hline
$C_2$ & $Y$-coordinate of the tail in $\mathcal{C}$ = $Y$-coordinate of $t_{target}$ in $\mathcal{C}_{target}$ \\ 
\hline
$C_3$ & $n \geq $max$\{M, m\} + 2$  \\ 
\hline
$C_4$ & $n \geq 2\cdot$max$\{N, H\}$, where $H$ is the length of the horizontal side of the smallest enclosing rectangle of $\mathcal{C'}$  \\ 
\hline
$C_5$ & The head in $\mathcal{C}$ is at the origin  \\ 
\hline
$C_6$ & $m \geq $ max$\{M, V\} + 1$, where $V$ is the length of the vertical side of the smallest enclosing rectangle of $\mathcal{C'}$  \\ 
\hline
$C_7$ & $\mathcal{C''} = \mathcal{C''}_{target}$  \\ 
\hline
$C_8$ & $\mathcal{C}'$ has a non-trivial reflectional symmetry with respect to a horizontal line \\ 
\hline 
\end{tabular}

\caption{The Boolean variable on the left is true if and only if the condition on the right is satisfied.}
\label{varbool}
\end{figure}

 We can logically divide the algorithm into seven phases. The starting configuration of the robots can fall into any one of the phases. These phases will be described in detail in the following subsections. Since the robots are oblivious, in each \textsc{Look-Compute-Move} cycle, it has to infer from the perceived configuration, which phase it is currently in. It does so by checking if certain conditions are fulfilled or not. These conditions can be expressed in terms of Boolean variables listed in Fig. \ref{varbool}.

The main algorithmic difficulty of the problem arises from the restrictions imposed on the movements of the robots. In the continuous setting the robots can freely move in any direction by arbitrarily small amounts and in some models, along any curve. Therefore, techniques used in the previous works on continuous space (e.g., \cite{Dieudonne10,Flocchini08,Petit09}) are not immediately portable in the discrete setting. Collision less movement is a major challenge in the grid model due to movement restriction. To resolve this, the tail expands the initial smallest enclosing rectangle (in Phase 1 and 3) making enough room for the interior robots to reconfigure themselves inside the rectangle without colliding. Our main idea is to utilize the asymmetry of the configuration to reach an agreement on a coordinate system, and try to keep the coordinate system invariant during the movements. To achieve this, in the first three phases, the head is put at the origin and the smallest enclosing rectangle is large enough so that the interior robots are confined in an appropriately small finite subgrid. Any movement by the interior robots restricted inside the finite subgrid keeps the coordinate system unaltered. So in Phase 4, the interior robots will rearrange themselves inside the finite subgrid to partially form the given pattern. In the final three phases, the head and the tail will move to their prescribed positions. Despite the apparent simplicity of the final three phases, designing movements is somewhat complicated as the coordinate system may change or the agreement in the coordinate system may be lost in some cases in the final phases.

\subsection{Phase 1}  

A robot infers that it is in Phase 1 if $\lnot(C_1 \wedge C_2) \wedge \lnot(C_3 \wedge C_4)$ is true \footnote{$A \wedge B$ is true if and only if both $A$ and $B$ are true. $\lnot A$ is true if and only if $A$ is false.}. In this case, the tail will move to the right and all other robots will remain static. Our aim is to make both $C_3$ and $C_4$ true.

\begin{theorem}
 If we have an asymmetric configuration $\mathcal{C}$ in phase 1 at some time $t$, then 
 
 \begin{enumerate}

  \item after one move by the tail towards right, the new configuration is still asymmetric and the coordinate system remains unchanged.
  
  \item after one move by the tail towards right, we have $\lnot(C_1 \wedge C_2) =$ true.
  
  \item after finite number of moves by the tail, we shall have $(C_3 \wedge C_4) =$ true.
  
  \item after finite number of moves by the tail, Phase 1 completes with $\lnot(C_1 \wedge C_2) \wedge (C_3 \wedge C_4) =$ true.
 \end{enumerate}

\end{theorem}

\begin{proof}

 1) Let the smallest enclosing rectangle at time $t$ be $\mathcal{R}^{old} = ABCD$, with $|AD| = m$ and $|AB| = n$, $n \geq m$. Since $\mathcal{C}$ is asymmetric, we must have $n > 2$. Let $\lambda_{AD}$ be the lexicographically largest string. Hence the tail $r$ is on the edge $BC$. After a move by the tail towards right, the smallest enclosing rectangle becomes $\mathcal{R}^{new} = AB'C'D$, with $r$ now on the edge $B'C'$.  
 
 We already have $\lambda_{AD}^{old} > \lambda_{DA}^{old}$. Suppose that in $\mathcal{C}$, $r$ is the only robot on $BC$. Then it is easy to see that $\lambda_{AD}^{new} > \lambda_{DA}^{new}$. Now assume that there are multiple robots on $BC$. Suppose that $r$ corresponds to the $x$th and $y$th term in $\lambda_{AD}^{old}$ and $\lambda_{DA}^{old}$ respectively.

 \textbf{Case-1 :} Let $x = y$. This means that $m$ is odd and $r$ is on the middle point of $BC$. Since the $x$th term is the last non-zero term in $\lambda_{AD}^{old}$, we must have $\lambda_{AD}^{old}\,_{\mbox{$\vert_{x-1}$}} > \lambda_{DA}^{old}\,_{\mbox{$\vert_{x-1}$}}$, where $\lambda\,_{\mbox{$\vert_{p}$}}$ is the string obtained by taking the first $p$ terms of $\lambda$. Therefore, $\lambda_{AD}^{new}\,_{\mbox{$\vert_{x-1}$}} > \lambda_{DA}^{new}\,_{\mbox{$\vert_{x-1}$}}$, which implies $\lambda_{AD}^{new} > \lambda_{DA}^{new}$. 
 
 \textbf{Case-2 :} Let $x < y$. Let $x = m(n-1)+b$. Since $x < y$, we have $b \leq \lfloor \frac{n}{2} \rfloor$. Since $r$ is the tail, the $(mn-m+1)$th to $(mn-m+b)$th terms of $\lambda_{DA}^{old}$ are all 0. The same is true for $\lambda_{DA}^{new}$. Since there were multiple robots on $BC$ at time $t$, there is at least one 1 among the $(mn-m+1)$th to $(mn-m+b)$th terms in $\lambda_{AD}^{new}$. Hence, $\lambda_{AD}^{new} > \lambda_{DA}^{new}$. 
 
  \textbf{Case-3 :} If $x > y$, then $r$ is encountered strictly earlier, as we scan $\mathcal{R}^{old}$ according to the string $\lambda_{DA}^{old}$, than $\lambda_{AD}^{old}$. This implies that $\lambda_{DA}^{old}\,_{\mbox{$\vert_{y}$}} > \lambda_{DA}^{new}\,_{\mbox{$\vert_{y}$}}$ and $\lambda_{AD}^{old}\,_{\mbox{$\vert_{y}$}} = \lambda_{AD}^{new}\,_{\mbox{$\vert_{y}$}}$. But as $\lambda_{AD}^{old} > \lambda_{DA}^{old}$, we have $\lambda_{AD}^{old}\,_{\mbox{$\vert_{y}$}} \geq \lambda_{DA}^{old}\,_{\mbox{$\vert_{y}$}}$. Hence, we have $\lambda_{AD}^{new}\,_{\mbox{$\vert_{y}$}} > \lambda_{DA}^{new}\,_{\mbox{$\vert_{y}$}}$, and so $\lambda_{AD}^{new} > \lambda_{DA}^{new}$.

 Therefore, we have $\lambda_{AD}^{new} > \lambda_{DA}^{new}$. Now we compare $\lambda_{AD}^{new}$ and $\lambda_{C'B'}^{new}$. Clearly, $r$ corresponds to the first 1 in $\lambda_{CB}^{old}$ (and also $\lambda_{C'B'}^{new}$). Suppose it is the $x$th term $\lambda_{CB}^{old}$. If the first 1 in $\lambda_{AD}^{old}$ appears before the $x$th term, then we easily have $\lambda_{AD}^{new} > \lambda_{C'B'}^{new}$. So let the $x$th term be the first 1 in $\lambda_{AD}^{old}$. Let the second 1 in $\lambda_{AD}^{old}$ be the $y$th term. Again, if the second 1 in $\lambda_{CB}^{old}$ appears beyond the $y$th term, then we are done. So assume that the second 1 in $\lambda_{CB}^{old}$ is the $y$th term. Hence, $\lambda_{AD}^{old}\,_{\mbox{$\vert_{y}$}} = \lambda_{CB}^{old}\,_{\mbox{$\vert_{y}$}}$. Now $\lambda_{C'B'}^{new}$ is obtained by inserting a string of 0's of length $m$ in $\lambda_{CB}^{old}$ after the $x$th term. Hence, we clearly have $\lambda_{AD}^{new}\,_{\mbox{$\vert_{y}$}} = \lambda_{AD}^{old}\,_{\mbox{$\vert_{y}$}} > \lambda_{C'B'}^{new}\,_{\mbox{$\vert_{y}$}}$. Thus, $\lambda_{AD}^{new} > \lambda_{C'B'}^{new}$. We can similarly show that $\lambda_{AD}^{new} > \lambda_{B'C'}^{new}$.

 Hence, $\lambda_{AD}^{new}$ is lexicographically strictly larger than $\lambda_{DA}^{new}$, $\lambda_{B'C'}^{new}$ and $\lambda_{C'B'}^{new}$. Since, $\mathcal{R}^{new} = AB'C'D$ is a non-square rectangle, we only have these four binary strings to consider. Hence, we can clearly see that the new configuration is still asymmetric and the coordinate system is unchanged. 
 
 2) Since only the tail moves, the value of $C_1$ remains unchanged. As the tail moves to the right, its $Y$-coordinate and hence $C_2$ is also unchanged. Therefore, $\lnot(C_1 \wedge C_2)$ remains true after the move. 
 
 3) It follows from 1) that the tail remains invariant throughout phase 1. Clearly, after a finite number of moves by the tail towards right, both $C_3$ and $C_4$ will become true.
 
 4) Follows from 2) and 3). \qed
\end{proof}

\subsection{Phase 2}

The algorithm is in phase 2, when  either $C_3 \wedge C_4 \wedge \lnot C_5 \wedge \lnot C_7$ or $\lnot C_2 \wedge C_3 \wedge C_4 \wedge \lnot C_5 \wedge C_7$ is true. Our aim is to take the head to the origin. Hence, in this phase, the head will move down towards the origin.

\begin{theorem}
 If we have an asymmetric configuration $\mathcal{C}$ in phase 2 at some time $t$, then 
 
 \begin{enumerate}
    
  \item after one move by the head downwards, the new configuration is still asymmetric and the coordinate system is unchanged.

  \item after finite number of moves by the head, $C_5$ becomes true.
  
  \item after finite number of moves by the head, phase 2 completes with $C_3 \wedge C_4 \wedge C_5 \wedge \lnot C_7$ or $\lnot C_2 \wedge C_3 \wedge C_4 \wedge C_5 \wedge C_7$ true.
 \end{enumerate}

\end{theorem}

\subsection{Phase 3}

The algorithm is in phase 3, if $C_3 \wedge C_4 \wedge C_5 \wedge \lnot C_6 \wedge \lnot C_7$ is true. In this phase, there are two cases to consider. The robots will check if $C_8$ is true or false. Let us first consider the case where $C_8$ is false. In this case, the tail will move upwards and the rest will remain static.

\begin{theorem}\label{p3}
 If we have an asymmetric configuration $\mathcal{C}$ in phase 3 at some time $t$ with $C_8 =$ false, then 
 
 \begin{enumerate}

  \item after one move by the tail upwards, the new configuration is still asymmetric and the coordinate system is unchanged.
  
  \item after one move by the tail upwards, we still have $C_4 \wedge C_5 \wedge \lnot C_7 =$ true.

  \item after finite number of moves by the tail, we shall have $C_3 \wedge C_4 \wedge C_5 \wedge C_6 \wedge \lnot C_7 =$ true.
 \end{enumerate}

 \end{theorem}

 \begin{proof}
 
  Let the smallest enclosing rectangle at time $t$ be $\mathcal{R}^{old} = ABCD$, with $|AD| = m$ and $|AB| = n$, $n > m$. Let $\lambda_{AD}^{old}$ be the lexicographically largest string.

  \textbf{Case 1 :} Suppose the smallest enclosing rectangle remains unchanged after the move, i.e., $\mathcal{R}^{new} = ABCD$. Since $C_4 \wedge C_5$ is true, it is easy to see that $\lambda_{AD}^{new} > \lambda_{BC}^{new}$ and $\lambda_{AD}^{new} > \lambda_{CB}^{new}$. Finally, $\lambda_{AD}^{new} > \lambda_{DA}^{new}$ follows from the fact that $C_4 \wedge C_5 \wedge \lnot C_8$ is true. To see this, note that $n \geq 2H$ as $C_4$ is true. This implies that the first $p = m\lfloor \frac{n}{2} \rfloor$ terms of both $\lambda_{AD}^{old}$ and $\lambda_{DA}^{old}$ contains all terms corresponding to the robots in $\mathcal{C}'$. Since $C_8$ is false, we must have $\lambda_{AD}^{old}\,_{\mbox{$\vert_{p}$}} > \lambda_{DA}^{old}\,_{\mbox{$\vert_{p}$}}$, and hence  $\lambda_{AD}^{new}\,_{\mbox{$\vert_{p}$}} > \lambda_{DA}^{new}\,_{\mbox{$\vert_{p}$}}$. So, $\lambda_{AD}^{new} > \lambda_{DA}^{new}$.

  \textbf{Case 2 :} Suppose that the tail is at $C$ at time $t$, and after a move upwards, we have $\mathcal{R}^{new} = ABC'D'$, $|AD'| = m+1$. Since $n \geq m+2 \Rightarrow n > m+1$ $\Rightarrow |AB| >|AD'|$, the smallest enclosing rectangle is still not a square. This implies that we still need to consider only four binary strings and it is easy to see that $\lambda_{AD'}$ is strictly largest among them.

 Therefore, we have proved 1). It is easy to see 2), i.e.,  $C_4 \wedge C_5 \wedge \lnot C_7$ is true after the move. However, $C_3$ might become false after the move described in case 2.  Then the phase changes to phase 1. Since, before the move, we had $C_3 =$ true, $n \geq m + 2$ and $n \geq M + 2$. If $C_3$ becomes false after the move, we have $n < $max$\{M, m + 1\} + 2$ $\Rightarrow n \leq m + 2$ (as $n \geq M + 2$). Since we also have $n \geq m + 2$, we get $n = m + 2 \Rightarrow m = n - 2$.  If $M \leq m$, after the move, $C_6$ becomes true as $m+1 \geq V+1$ and $m+1 \geq M+1$. If $m < M$, then $n-2 < M \Rightarrow n < M + 2$, a contradiction. Therefore, after the move, we have $C_6 =$ true, i.e., we have $\lnot C_3 \wedge C_4 \wedge C_5 \wedge C_6 \wedge \lnot C_7 =$ true. So we are in phase 1, and the tail will move rightwards. After one rightwards move, we shall have $C_3 \wedge C_4 \wedge C_5 \wedge C_6 \wedge \lnot C_7 =$  true. \qed

 \end{proof}

\begin{figure}[h]
\centering
\subcaptionbox[Short Subcaption]{ Case 1
      \label{phase3}
}
[
    0.8\textwidth 
]
{
    \fontsize{9pt}{9pt}\selectfont
    \def\svgwidth{0.4\textwidth}
    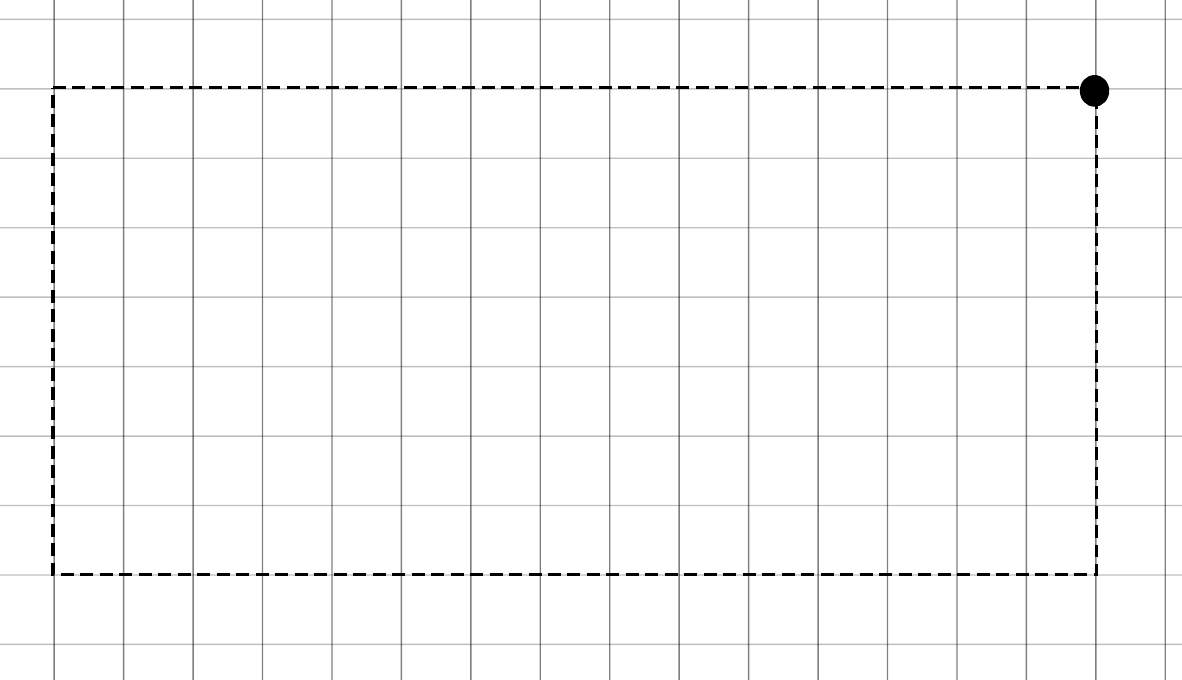
}
 \\
\subcaptionbox[Short Subcaption]{ Case 2
      \label{phase3b}
}
[
    0.8\textwidth 
]
{
    \fontsize{9pt}{9pt}\selectfont
    \def\svgwidth{0.4\textwidth}
    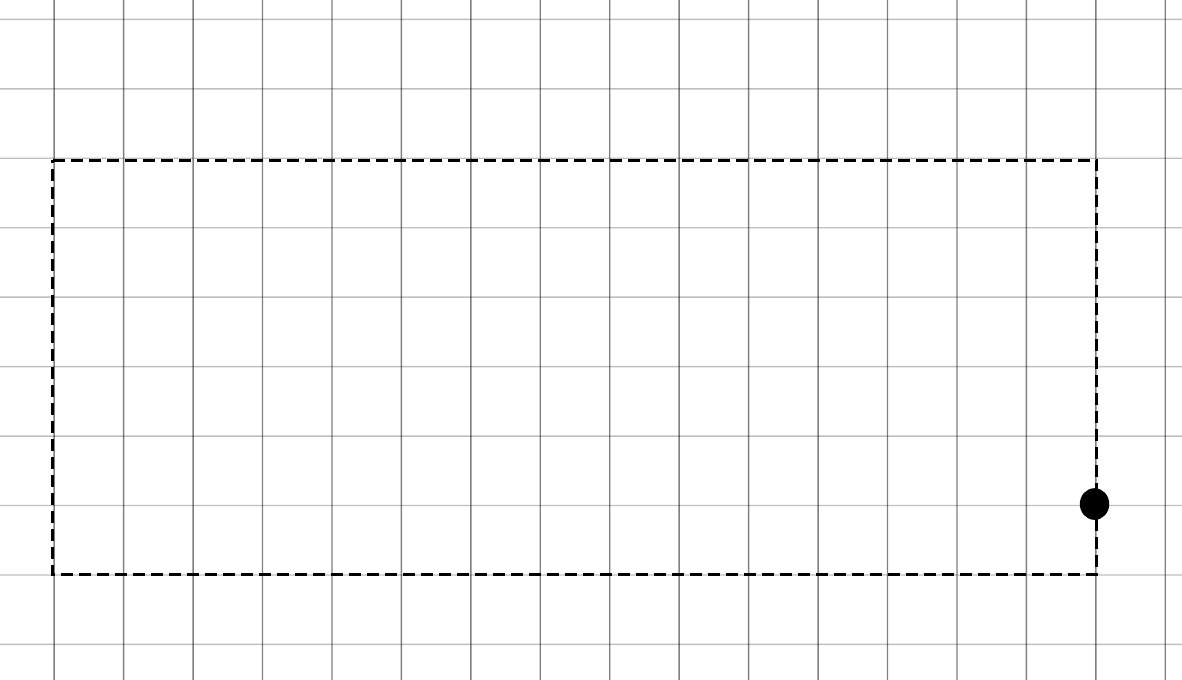
}

\caption[Short Caption]{Illustration of phase 3 with $C_8 =$ true.}

\end{figure}

 Now assume that $C_8$ is true, i.e., $\mathcal{C}'$ has a non-trivial reflectional symmetry with respect to a horizontal line $L$. Again, let the smallest enclosing rectangle be $\mathcal{R}^{old} = ABCD$, with $|AD| = m$ and $|AB| = n$, $n > m$. Let $\lambda_{AD}^{old}$ be the lexicographically largest string. Let $E$ be the point of $BC$ where it intersects with $L$. Let the smallest enclosing rectangle of $\mathcal{C}'$ be $\mathcal{R}' = AB'C'D'$. There are two cases to consider: $D \neq D'$ (Case 1) and $D = D'$ (Case 2).
 
 \textbf{Case 1:} See Fig.\ref{phase3}. In this case, the tail $r$ will move upwards. Clearly, after finite number of moves by the tail $C_6 =$ true is achieved. However, $C_3$ may become false, but is recovered after one move rightwards as explained previously.
 
 \textbf{Case 2:} See Fig.\ref{phase3b}. Since the configuration is asymmetric and $\lambda_{AD}^{old}$ is the largest string, the tail $r$ must be in $[B,E)$. In this case, $r$ will move downwards. When $r$ goes below $B$, the coordinate system flips. The new coordinate system has origin at $D$, $X$-axis $= \overrightarrow{DC}$ and $Y$-axis $= \overrightarrow{DA}$. Clearly, the case is reduced to the situation similar to case 1.
 
 \begin{theorem}
  If we have an asymmetric configuration $\mathcal{C}$ in phase 3 at some time $t$ with $C_8$ true, then after finite number of moves by the tail, we have $C_3 \wedge C_4 \wedge C_5 \wedge C_6 \wedge \lnot C_7 =$ true.
 \end{theorem}

\subsection{Phase 4}

If the configuration satisfies $C_3 \wedge C_4 \wedge C_5 \wedge C_6 \wedge \lnot C_7 =$ true, then the algorithm is in phase 4. In this phase, the head and the tail will remain static. Let $\mathcal{F}$ be the subgrid of $\mathcal{R}$ of size $(m-1) \times \lfloor \frac{n}{2} \rfloor $ with coinciding bottom-left corners (See Fig. \ref{phase4}). $\mathcal{F}$ can be considered as a finite line segment $\mathcal{L}$ as shown in Fig. \ref{phase4}. The interior robots will execute the protocol \textsc{PFonPath()} on $\mathcal{L}$ to achieve $\mathcal{C''} = \mathcal{C''}_{target}$, i.e., $C_7 =$ true.

 \begin{figure}[h]
\centering
    \fontsize{8pt}{8pt}\selectfont
    \def\svgwidth{0.5\textwidth}
    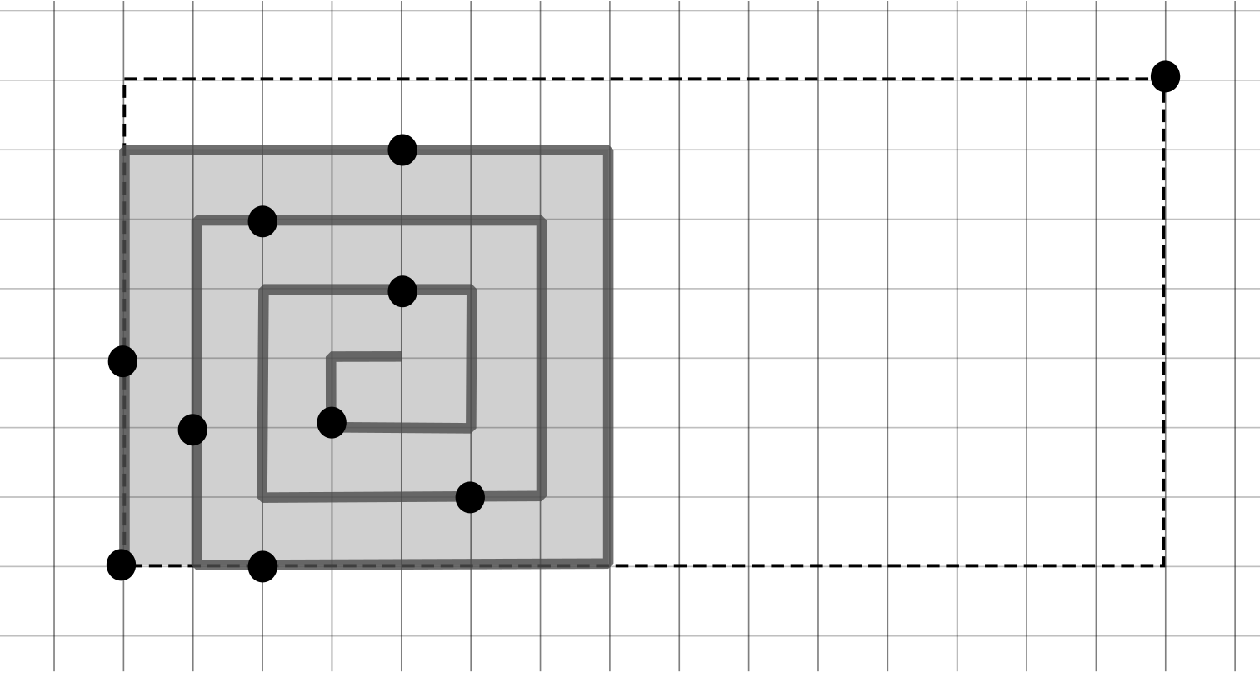

\caption{A configuration in phase 4.}
\label{phase4}
\end{figure}

\begin{theorem}
 If we have an asymmetric configuration $\mathcal{C}$ at some time $t$, with $C_3 \wedge C_4 \wedge C_5 \wedge C_6 \wedge \lnot C_7 =$ true, then 
 
 \begin{enumerate}
    
  \item after any move by an interior robot according to \textsc{PFonPath()}, the new configuration is still asymmetric and the coordinate system is unchanged.
  
  \item  after any move by an interior robot according to \textsc{PFonPath()}, we still have $C_3 \wedge C_4 \wedge C_5 \wedge C_6$ $=$ true.
  
  \item after finite number of moves by the interior robots, we shall have $C_7$ $=$ true.
  
  \item after finite number of moves by the interior robots, Phase 4 completes with $C_3 \wedge C_4 \wedge C_5 \wedge C_6 \wedge C_7$ $=$ true.
 \end{enumerate}

\end{theorem}

\begin{proof}
 It is easy to see that $C_3 \wedge C_4 \wedge C_5 \wedge C_6$ is true after any move. Let $\mathcal{R}^{old} = \mathcal{R}^{new} = ABCD$. Let $\lambda_{AD}^{old}$ be the lexicographically largest string in $\mathcal{R}^{old}$. $C_5 =$ true implies that $A$ is occupied. $C_4 \wedge C_6 =$ true implies that $B, D$ are empty and $C$ is occupied by the tail. Hence, $\lambda_{AD}^{new} > \lambda_{BC}^{new}$ and $\lambda_{AD}^{new} > \lambda_{DA}^{new}$. Also, it implies from $C_4 =$ true that $\lambda_{AD}^{new} > \lambda_{CB}^{new}$. The rest are easily seen. \qed
\end{proof}

\subsection{Phase 5}

The algorithm is in phase 5, if $\lnot C_2 \wedge C_3 \wedge C_4 \wedge C_5  \wedge C_7$ is true. In this phase, the tail will move along the vertical grid line in order to make $C_2$ true. 

If we have an asymmetric configuration $\mathcal{C}$ which is in phase 5, then depending on whether $C_8$ is true or false, there are two cases to consider. 

Let the smallest enclosing rectangle of $\mathcal{C}$ be $\mathcal{R} = ABCD$, with $|AD| = m$ and $|AB| = n$, $n > m$. Let $A$ be the leading corner, and hence $X$-axis $= \overrightarrow{AB}$ and $Y$-axis $= \overrightarrow{AD}$. Now, let us plot the points of $\mathcal{C}_{target}$ in this coordinate system. Except $h_{target}$ and $t_{target}$, all other points of $\mathcal{C}_{target}$ are occupied by the robots. Let the smallest enclosing rectangle of $\mathcal{C}'$ be $\mathcal{R}' = AB'C'D'$ (See Fig. \ref{phase5}). Hence, the tail is currently on $BC$, and all the remaining robots are inside the region $AB'C'D'$.

\textbf{Case-1} First, assume that $C_8$ is true. Since the head is at $A$, $D'$ is also occupied due to the symmetry.  Let $C''$ be the grid point where the grid lines $\overleftrightarrow{D'C'}$ and $\overleftrightarrow{BC}$ intersect. Let $C'''$ be the middle point of $BC''$. $C'''$ is a grid point if $|BC''|$ is odd. If the tail is on $BC''$, then $C = C''$ and $D = D'$. Note that the tail can not be on $[C''',C'']$, because then we shall have $\lambda_{DA} \geq \lambda_{AD}$. Hence, the tail is on $[B,C''')$ or $(C'',\infty)$.

In this phase, we want to make $C_2$ true. This means that the tail needs to go to the grid point on $\overleftrightarrow{BC''}$ that is on the same horizontal line with  $t_{target}$. Let us call this point $\tilde{t}_{target}$. Consider the case where $\tilde{t}_{target} \in [B,C'']$. In this case, the upper left corner of the smallest enclosing rectangle of $\mathcal{C}_{target}$ is $D'$, which is occupied by a robot. Since the input is given in canonical coordinates, the bottom left corner (origin) of the smallest enclosing rectangle of $\mathcal{C}_{target}$, i.e., $A$, must be the leading corner. Therefore, $A$ must be occupied in the final configuration. Since $A$ is already occupied, it implies that $C_1$ is currently true. Also note that $\tilde{t}_{target} \notin (C''',C'']$, as $A$ is the leading corner in $\mathcal{C}_{target}$. Hence, $\tilde{t}_{target}$ is on $[B,C''']$ or $(C'',\infty)$.


\underline{\textbf{Case 1A: tail} $\pmb{\in [B,C''')}$ \textbf{and}  $\pmb{\tilde{t}_{target} \in [B,C''']}$}

The tail will move towards $\tilde{t}_{target}$. During the movements, the coordinate system remains invariant. However, if $\tilde{t}_{target}$ is at $C'''$, a horizontal symmetry will be created when it reaches $\tilde{t}_{target}$.

\underline{\textbf{Case 1B: tail} $\pmb{\in (C'',\infty)}$ \textbf{and}  $\pmb{\tilde{t}_{target} \in (C'',\infty)}$}

The tail will move towards $\tilde{t}_{target}$. Again it is easy to see that the coordinate system remains invariant during the movements.

\underline{\textbf{Case 1C: tail} $\pmb{\in (C'',\infty)}$ \textbf{and}  $\pmb{\tilde{t}_{target} \in [B,C''']}$}

In this case, the tail will move downwards. When $r$ reaches $C''$, the coordinate system flips. The new coordinate system has origin at $D'$, $X$-axis $= \overrightarrow{D'C''}$ and $Y$-axis $= \overrightarrow{D'A}$. In the new coordinate system, $r$ requires to place itself in $[C'',C''']$. Hence, the case is reduced to the situation similar to case 1A. Thus $r$ achieves $C_2 =$ true without going beneath $C'''$.

\underline{\textbf{Case 1D: tail} $\pmb{\in [B,C''')}$ \textbf{and}  $\pmb{\tilde{t}_{target} \in (C'',\infty)}$}

In this case, the tail will move downwards. When $r$ goes beneath $B$, the coordinate system flips. The new coordinate system has origin at $D'$, $X$-axis $= \overrightarrow{D'C''}$ and $Y$-axis $= \overrightarrow{D'A}$. Clearly, the case is reduced to the situation similar to case 1B.

%


\textbf{Case-2} Now assume that $C_8$ is false. It is easy to see that where ever $\tilde{t}_{target}$ is on $\overleftrightarrow{BC}$, the binary string attached to $A$ is lexicographically strictly largest as the tail moves towards it. Hence, the movement of the tail in phase 5 does not change the coordinate system. Clearly, after finite number of moves, $C_2$ becomes true. Then phase 5 is completed with  $C_2 \wedge C_3 \wedge C_4 \wedge C_5 \wedge C_7$ true.

 \begin{figure}[h]
\centering
    \fontsize{8pt}{8pt}\selectfont
    \def\svgwidth{0.45\textwidth}
    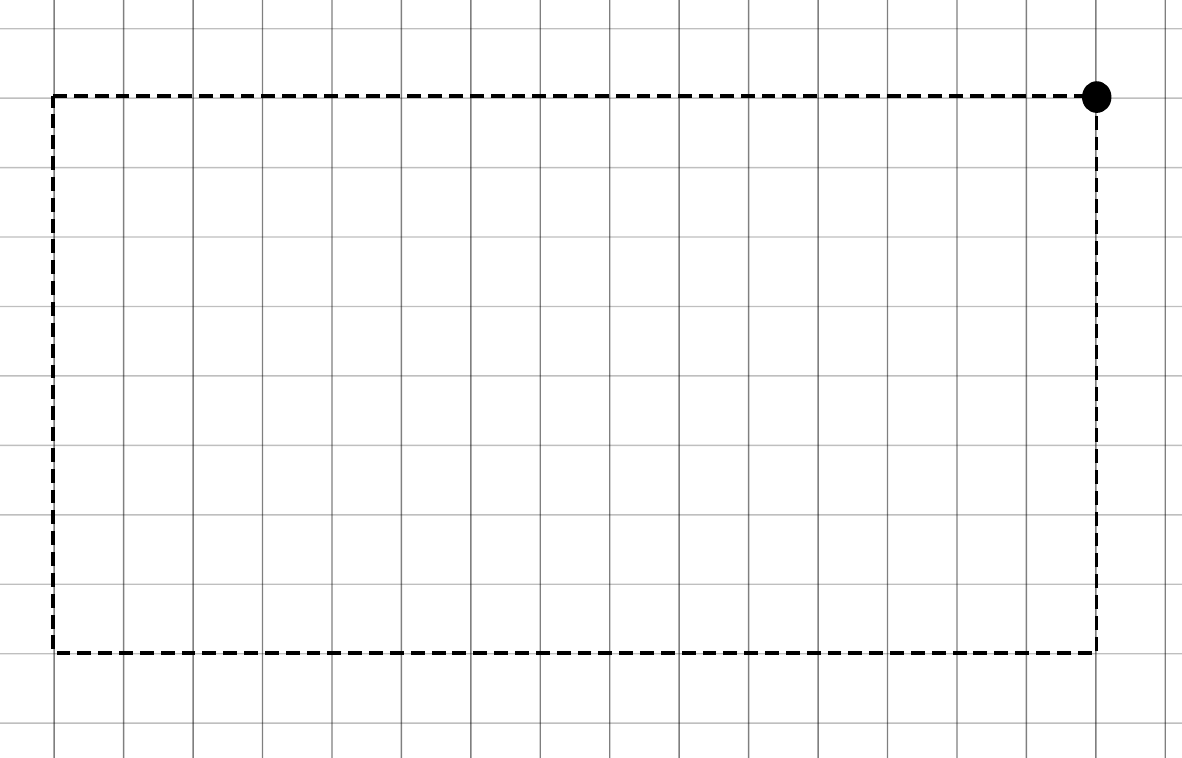

\caption{Illustration of case 1 of phase 5.}
\label{phase5}
\end{figure}

%

\begin{theorem}
 If we have an asymmetric configuration $\mathcal{C}$ in phase 5 at some time $t$, then after finite number of steps phase 5 completes with $C_2 \wedge C_3 \wedge C_4 \wedge C_5 \wedge C_7 =$ true.
 
%
%
%
%

\end{theorem}

%
%
%
%
%

\subsection{Phase 6}

If we have $\lnot C_1 \wedge C_2 \wedge C_3 \wedge C_4 \wedge C_7 =$ true, then the algorithm is in phase 6. In this case, the head will move towards $h_{target}$. 

 \begin{figure}[h]
\centering
    \fontsize{8pt}{8pt}\selectfont
    \def\svgwidth{0.45\textwidth}
    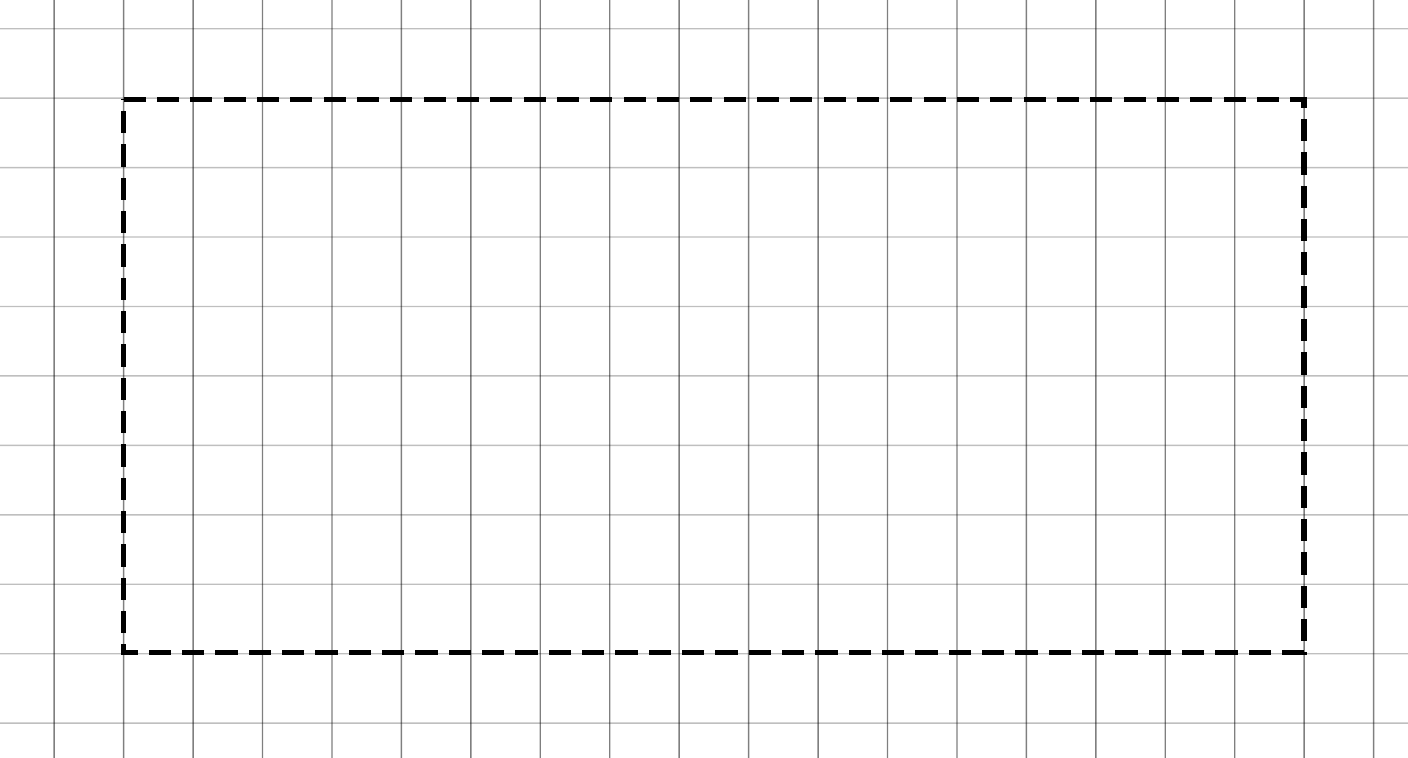

\caption{Illustration of phase 6.}
\label{phasefinal}
\end{figure}

Consider an asymmetric configuration $\mathcal{C}$ which is in phase 6. Let $ABCD$ be the smallest enclosing rectangle, with $\lambda_{AD}^{}$ being the (strictly) largest string. Let $H$ and $T$ be the position of the head and the tail respectively. $H$ and $T$ are clearly on $AD$ and $BC$ respectively. Plot the points of $\mathcal{C}_{target}$ on the grid with respect to the current coordinate system ($X$-axis $= \overrightarrow{AB}$ and $Y$-axis $= \overrightarrow{AD}$). The smallest enclosing rectangle of these points is $AB'C'D$ (See Fig. \ref{phasefinal}). Let $H'$ and $T'$ be the points $h_{target}$ and $t_{target}$. Therefore, if the head moves from $H$ to $H'$ and the tail moves from $T$ to $T'$, then the given pattern is formed. $H'$ and $T'$ are clearly on $AD$ and $B'C'$ respectively, with $T$ and $T'$ being on the same horizontal line.

The aim of this phase is to move the head from $H$ to $H'$. Let $\mathcal{C}_{h}$ be the configuration obtained from $\mathcal{C}$, if the the head moves from $H$ to $H'$. Since the input is given in canonical coordinates, $\lambda_{AD}^{}$ is the lexicographically largest string in $\mathcal{C}_{target}$. This implies that $\lambda_{AD}^{}$ is lexicographically largest in $\mathcal{C}_{h}$. Recall that, $\lambda_{AD}^{}$ may not be strictly largest, as $\mathcal{C}_{target}$ may have some symmetries. In particular, if $\mathcal{C}_{target}$ has a reflectional symmetry with respect to a horizontal axis, then $\mathcal{C}_{h}$ also has the same. Then $\lambda_{AD}^{}$ and $\lambda_{DA}^{}$ are both lexicographically largest in $\mathcal{C}_{h}$. In any case, as the head moves from $H$ to $H'$, $\lambda_{AD}^{}$ remains lexicographically strictly largest until it reaches $H'$.

\begin{theorem}
 Let $\mathcal{C}$ be an asymmetric configuration in phase 6. Then after finite number of moves by the head, phase 6 completes with $C_1 \wedge C_2 \wedge C_3 \wedge C_4 =$ true, and hence, $\lnot C_0 \wedge C_1 \wedge C_2 =$ true. The final configuration has a horizontal reflectional symmetry, if $\mathcal{C}_{target}$ also has the same.

\end{theorem}

\subsection{Phase 7}

Finally, the algorithm is in phase 7 if we have $\lnot C_0 \wedge C_1 \wedge C_2 =$ true. Suppose that the tail is at $T'$. The tail will move horizontally towards $T = t_{target}$. Note that a configuration $\mathcal{C}$ in phase 7 can have a reflectional symmetry with respect to a horizontal axis. Let $\lambda_{AD}$ be a largest string in the configuration. Since $\lambda_{AD}$ is also a largest string in the final configuration, it is easy to see that, when the tail is moving from $T'$ to $T$, $\lambda_{AD}$ remains the largest string (may be jointly with $\lambda_{DA}^{}$) until it reaches $T$. 


\begin{theorem}
 If we have a configuration $\mathcal{C}$ at some time $t$, with $\lnot C_0 \wedge C_1 \wedge C_2 =$ true, then after finite number of moves by the tail, $C_0$ becomes true.
 
\end{theorem}

It is not difficult to verify (See Appendix) that any configuration with $C_0 =$ false, belongs to one of the seven phases that we have discussed. From the results we have proved, it follows that starting from any asymmetric configuration, our algorithm can form any given pattern in finite time (See Appendix for a phase transition diagram of our proposed algorithm). Hence, we can conclude the following theorem.

\begin{theorem}
 \textsc{Arbitrary Pattern Formation} is solvable in ASYNC from any asymmetric initial configuration.
\end{theorem}

\section{Concluding Remarks}\label{conclu}

We have proved that any arbitrary pattern is formable by a set of asynchronous robots if the initial configuration is asymmetric. The immediate course of future research would be to characterize the patterns formable from symmetric configurations. It can be proved that if a configuration $\mathcal{C}$ admits symmetry $\varphi$ such that no robot lies on the axis of reflection or the center of rotation, then any configuration formable from $\mathcal{C}$ necessarily has the same symmetry $\varphi$. This is however not true, if the axis of reflection or the center of rotation contains a robot $r$. The symmetry may be broken by asking the robot $r$ to move. However, this is not straightforward especially in a crowded situation. It would be also interesting to consider randomized algorithms. Another direction of future research would be to extend our work for patterns allowing multiplicities. 

\paragraph{Acknowledgements.} The first three authors are supported by NBHM, DAE, Govt. of India, CSIR, Govt. of India and UGC, Govt. of India respectively. We would like to thank the anonymous reviewers for their valuable
comments which helped us improve the quality and presentation of this paper.

\bibliographystyle{splncs04}
\bibliography{pattern_grid}

\newpage
\appendix

\newpage

\section*{Appendix}

\section{Different Phases of the Main Algorithm}

 From the following tree, it can be easily seen that any configuration with $C_0 =$ false belongs to one of the seven phases described in the algorithm. It is also evident that the phases are mutually disjoint. 
 
 \bigbreak
 
 \bigbreak
 
\begin{forest}
 for tree={l'=0pt,
    draw},forked edges,
[$\lnot C_0$, s sep = 2cm, 
[$\lnot C_0 \wedge C_1\wedge C_2$\\\normalsize (Phase 7), align=center, base=bottom], 
[$\lnot C_0 \wedge \lnot (C_1\wedge C_2)$\\$\Leftrightarrow \lnot (C_1\wedge C_2)$, align=center, base=bottom, s sep = 4cm, calign = last
[$\lnot (C_1\wedge C_2)\wedge C_3\wedge C_4$,  calign = first
[$\lnot (C_1\wedge C_2)\wedge C_3\wedge C_4\wedge C_7$,  calign = first, l = 3.5cm 
[$\lnot C_2\wedge C_3\wedge C_4\wedge C_7$, calign = first, l = .5cm
[$\lnot C_2\wedge C_3\wedge C_4\wedge C_5\wedge C_7$\\\normalsize (Phase 5), align=center, base=bottom]
[$\lnot C_2\wedge C_3\wedge C_4\wedge\lnot C_5\wedge C_7$\\\normalsize (Phase 2), align=center, base=bottom]
]
[$\lnot C_1\wedge C_2\wedge C_3\wedge C_4\wedge C_7$\\\normalsize (Phase 6), align=center, base=bottom],l = 4cm
]
[$\lnot (C_1\wedge C_2)\wedge C_3\wedge C_4\wedge\lnot C_7$\\$\Leftrightarrow C_3\wedge C_4\wedge\lnot C_7$, align=center, base=bottom, calign = last
[$C_3\wedge C_4\wedge C_5\wedge\lnot C_7$
[$C_3\wedge C_4\wedge C_5\wedge C_6\wedge\lnot C_7$\\\normalsize (Phase 4), align=center, base=bottom]
[$C_3\wedge C_4\wedge C_5\wedge\lnot C_6\wedge\lnot C_7$\\\normalsize (Phase 3), align=center, base=bottom]
]
[$C_3\wedge C_4\wedge\lnot C_5\wedge\lnot C_7$\\\normalsize (Phase 2), align=center, base=bottom]
]
]
[$\lnot (C_1\wedge C_2)\wedge\lnot (C_3\wedge C_4)$\\\normalsize (Phase 1), align=center, base=bottom]
]
]
\end{forest}

 \newpage
 
 \section{Phase Transition Diagram of the Main Algorithm}

A general scheme of transition between different phases of the algorithm is shown in the following diagram. Observe that the only cycle in the graph is the one involving phase 1 and phase 3. It has been shown in the proof of Theorem \ref{p3} that this does not create a livelock. Clearly, starting from any phase, the algorithm terminates with $C_0 =$ true.

 \begin{figure}[h] 
\centering
\scalebox{0.8}{ 
\begin{tikzpicture}[shorten >=1pt,node distance=2.5cm,auto]
    \node[state]                    (1)                      {$Phase 1$} ++(-4cm,1cm);;
    \node[state]                    (2) [below of = 1] {$Phase 2$};
    \node[state]                    (3) [below right of = 2]       {$Phase 3$};
    \node[state]                    (4) [below of = 3] {$Phase 4$};
    \node[state]                    (5) [below right of = 4] {$Phase 5$};
    \node[state]                    (6) [below of = 5] {$Phase 6$};
    \node[state]                    (7) [below right of = 6] {$Phase 7$};
    \node[state,accepting]          (F) [below of = 7] {$C_0$};

    \draw[->] (1) edge				node {}	(2)
	      (1) edge	[bend left = 20]	node {}	(3)
	      (1) edge	[bend left = 45]	node {}	(4)
	      (1) edge	[bend left = 45]	node {}	(5)
	      (1) edge	[bend left = 60]	node {}	(6)
	      
	      (2) edge				node {}	(3)
	      (2) edge	[bend right = 20]	node {}	(4)
	      (2) edge	[bend right = 50]	node {}	(5)
	      
	      (3) edge				node {}	(4)
 	      (3) edge	[bend left = 90]	node {}	(1)

	      (4) edge				node {}	(5)
	      
	      (5) edge				node {}	(6)
	      (5) edge	[bend left = 20]	node {}	(7)
	      
	      (6) edge          		node {} (7)
	      
	      (7) edge          		node {} (F);

\end{tikzpicture}}
\end{figure}
 
\end{document}